\def\year{2019}\relax
\documentclass[letterpaper]{article} 
\usepackage{aaai19}  
\usepackage{times}  
\usepackage{helvet}  
\usepackage{courier}  
\usepackage{url}  
\usepackage{graphicx}  
\usepackage{xcolor}
\usepackage{soul}
\usepackage[utf8]{inputenc}
\usepackage[small]{caption}

\usepackage{etex}

\usepackage{url}
\usepackage{color}
\usepackage{graphicx}
\usepackage{amsthm}
\usepackage{amsfonts}
\usepackage{amsmath}

\usepackage{booktabs} 

\usepackage{color}

\newtheorem{theorem}{Theorem}
\newtheorem{observation}{Observation}
\newtheorem{definition}{Definition}
\newtheorem{corollary}{Corollary}
\newtheorem{lemma}{Lemma}
\newtheorem{example}{Example}

\usepackage{bbm} 
\definecolor{gray}{gray}{0.5}
\definecolor{darkgreen}{rgb}{0,0.5,0}
\newcommand{\mynote}[2]{\ifnum\Comments=1\textcolor{#1}{#2}\fi}
\newcommand{\reals}{\mathbb{R}}
\newcommand{\ones}{\mathbbm{1}}

\usepackage{tikz,pgfplots,tikz-3dplot}
\usetikzlibrary{arrows,shapes,positioning,calc,intersections,through}\tdplotsetmaincoords{55}{110}        
\tikzstyle{critical}=[solid]
\tikzstyle{indifference}=[dashed]
\tikzstyle{harmless}=[red,fill opacity=0.4]
\tikzstyle{fotakis}=[blue,fill opacity=0.4]
\tikzstyle{noitem}=[red,fill opacity=0.4]
\tikzstyle{item1}=[blue,fill opacity=0.4]
\tikzstyle{item2}=[green,fill opacity=0.4]
\tikzstyle{allNon}=[green,fill opacity=0.4]
\tikzstyle{restrictedNon}=[black,fill opacity=0.4]
\newcommand{\axes}{%
  \path (0,0) edge[->,>=latex] node[below] {$\theta_1'$} (2.1,0);
  \path (0,0) edge[->,>=latex] node[left] {$\theta_2'$} (0,2.2);  }
 \newcommand{\axesaxes}{%
  \path (0,0) edge[->,>=latex] node[below] {$\theta_1$} (2.1,0);
  \path (0,0) edge[->,>=latex] node[left] {$\theta_2$} (0,2.2);  }
\newcommand{\putpoint}[3]{%
  \fill (#1) circle [radius=1pt] node[#2] {$#3$};  }

\begin{document}
\title{Partial Verification as a Substitute for Money}
\author{Sofia Ceppi\\
PROWLER.io\\
sofia@prowler.io\\
\And
{\bf \Large Ian Kash}\\
University of Illinois at Chicago\\
iankash@uic.edu\\
\And
{\bf \Large Rafael Frongillo}\\
University of Colorado Boulder\\
raf@colorado.edu\\ 
}
\maketitle

\begin{abstract}
Recent work shows that we can use partial verification instead of money to implement truthful mechanisms.  In this paper we develop tools to answer the following question.  Given an allocation rule that can be made truthful with payments, what is the minimal verification needed to make it truthful without them?  Our techniques leverage the geometric relationship between the type space and the set of possible allocations.
\end{abstract}

\section{Introduction}

Mechanism design studies how to realize desirable outcomes to optimization problems in settings with self-interested agents.  The most common tool to achieve desirable outcomes is the use of payments, and there is a large literature focusing on the following question.  Given an allocation rule, which specifies the outcome that should be selected given the types of the agents, do there exist payments to turn it into an incentive compatible mechanism (to {\em implement} it)~\cite{guesnerie1984,saks2005,ashlagi2010,frongillo2014}?

Recent work has identified partial verification as a useful alternative to money to implement incentive compatible mechanisms. The idea is that the mechanism designer can detect some possible agent misreports, either by preventing them outright, or by penalizing the agent (e.g., by excluding her from the market). 
The power of partial verification is that the mechanism designer need not provide agents with incentives for a subset of the possible types they can report if verification of these types is in place.  An example of such a partial verification is agents not being able to report a higher valuation for any assignment than is true.  They are free, however, to report a lower value.

This specific type of verification has been adopted by Fotakis, Krysta and Ventre~(\citeyear{fotakis2014}). They consider the case in which a government is auctioning business licenses for cities under its administration, and companies want to get a license for some subset of cities to sell their stock of goods to the market.  The verification assumption is that the government, that acts as auctioneer, can verify if the winner actually has sufficient goods in stock and thus prevent overbidding.
In a particularly suggestive result, they showed that this verification suffices to implement all allocation rules that are implementable with money in single-minded combinatorial auctions.
(In fact they show that a weaker verification, where the agent cannot overbid only on the bundle received, suffices.)
Interestingly, this verification no longer suffices for agents who are $k$-minded, $k \geq 2$.  Other work in the literature (see Related Work)
has focused on specific scenarios like facility location and combinatorial auctions, and identifies the sets of verification assumptions that guarantee incentive compatibility of mechanisms.

Our work takes a different approach, in that we build tools to understand the power of verification independent of a specific scenario.
In particular, we answer the following question: given an allocation rule that can be made truthful with payments, what is the minimal verification needed to make it truthful without them? Essentially, similar to~\cite{diodato2016}, we aim to inform the designer of the resources needed for verification.  In contrast to this work, which focuses on facility location, we propose a geometric characterization of the verification needed to use any implementable-with-payment allocation rule in a scenario without transfers, while guaranteeing strategyproofness.

We introduce the concept of the \emph{harmless set of types} as those which do not need to be verified for a given set of single-agent allocation rules.  Our basic building block is a characterization of the structure of harmless sets for allocation rules which only assign two possible allocations and are implementable with payments.
We then show how this can be extended to characterize harmless sets for more general sets of single-agent implementable-with-payments allocation rules.
Our focus on sets of rules derives from the observation that multiagent allocation rules are, from the perspective of a single agent, just a set of allocation rules parameterized by the types of the other agents.

Our characterization highlights a split in the nature of harmless sets of types for deterministic and randomized mechanisms.  Deterministic and universally truthful mechanisms both have large harmless sets of types, while in contrast the harmless set for truthful-in-expectation mechanisms is quite restricted.  Our results are constructive and provide geometric insights for these findings.

The central contribution of our approach is its generality: our analysis could in principle be applied to any mechanism or class of mechanisms. Moreover, while our results are often stated for two allocations and single-agent settings, they also apply directly to more than two allocations and multi-agent settings, by standard arguments.
We also examine two extensions: allocation-dependent verification, a weaker form of verification that can expand the harmless set of types, and reverse verification, where the reported type of the agent is considered when computing the types that needs to be verified.
 We conclude with examples showing how our approach can be used in three application domains and how our results replicate and extend existing results in the literature.

\section{Related Work}
\label{sec:relatedwork}

Several works in both the economics and computer science literatures focus on the design of incentive compatible mechanisms with verification~\cite{green1986,fotakis2015,fotakis2015b,penna2009,ventre2014} to overcome the Gibbard-Satterthwaite impossibility result~\cite{gibbard,satterthwaite} for mechanisms without money.

In particular, given a mechanism, these works aim to reduce the types agents can report to the ones that do not bring benefit to them. This is done by  either assuming that the reportable type space varies depending on the true type of the agents and that this is known to the mechanism~\cite{green1986} or that the mechanism can verify part of the type space and penalize agents that misreport in that space~\cite{fotakis2014,diodato2016}.
In particular, given a mechanism, these works aim to reduce the types agents can report to the ones that do not bring benefit to them. This is done by  either assuming that the reportable type space varies depending on the true type of the agents or that the mechanism can verify part of the type space and penalize agents that misreport in that space.

The power of verification in the design of mechanisms without money has been studied in a number of applications including scheduling of unrelated machines \cite{koutsoupias2014}, combinatorial auctions~\cite{krysta2015}, and assignment and allocation problems~\cite{dughmi2010,guo2010}. A particular focus has been on the design of mechanisms with verification yielding good approximate solutions to the problem of facility location on a line~\cite{procaccia2009,serafino2016,serafino2014,diodato2016}.
Among the work that study the power of verification, several focused on the case in which the mechanism is without money. The design of mechanism without money has been studied for different applications: from scheduling of unrelated machine \cite{koutsoupias2014} and combinatorial auctions~\cite{krysta2015} to assignment and allocation problems~\cite{dughmi2010,guo2010}.
Often, the focus has been on the design of approximated mechanisms with the aim to face situations in which optimal solutions do not exist. In particular, the most considered problem is the one of facilities location on a line~\cite{procaccia2009,serafino2016,serafino2014,diodato2016}.
Much of this literature has focused on identifying verifications which seem natural for a particular application and suffice to design useful mechanisms
~\cite{koutsoupias2014,serafino2016,fotakis2014}, in contrast to the present work which fixes a set of allocation rules and asks what verification would be minimally necessary to render the mechanisms truthful.
Most similarly to our own work, Ferraioli et al.~(\citeyear{diodato2016}) have considered the question the minimum set of assumptions needed in the facility location setting.

\section{Preliminaries}

In this section, we focus on mechanism design with a single agent.  Let $S$ denote the set of assignments, one of which the agent will receive, $|S| = m$.  Let $A$ denote the set of allocations, where an allocation $a  \in A$ is a probability distribution over assignments. Formally, $A \subseteq \{ a \in [0,1]^m : \sum_{s \in S} a(s) = 1\}$.  One can think of assignments as a set of mutually exclusive outcomes, and allocations as distributions over these outcomes, which for example will be point distributions when we consider deterministic mechanisms.

We use $\theta =  [\theta_{s_1}, \theta_{s_2}, ..., \theta_{s_m} ] \in \reals^m$ to denote the type of the agent, i.e., her private information, where $\theta_{s_i}$ is the agent's value for the assignment $s_i$.   The set of possible agent types is denoted $\Theta \subseteq \mathbb{R}^m$.
A single-agent direct revelation mechanism is denoted $\mathcal{M} = \{f, p\}$, where $f : \Theta \rightarrow A$ is the allocation rule and $p: \Theta \rightarrow \mathbb{R}$ is the payment rule.  Under this mechanism, the utility of an agent with type $\theta$ who reports type $\hat{\theta}$ is
 $u^{\mathcal{M}}(\theta, \hat{\theta}) = f(\hat{\theta}) \cdot \theta - p(\hat{\theta})$.

We now introduce several terms and definitions which we will use throughout the paper.  
Additionally, we summarize in Table~\ref{tab:notation} the notation used in the paper which will be introduced in the following sections.

\medskip\textbf{Incentive Compatibility.}~
A mechanism is incentive compatible (i.e., truthful) if the agent is incentivized to communicate to the mechanism her true type.  Since agents are rational, to guarantee incentive compatibility, the mechanism must guarantee that each agent is better off when she reports her true type than when she misreports, i.e., $u^{\mathcal{M}}(\theta, \theta) \geq u^{\mathcal{M}}(\theta, \hat{\theta})$ for all $\theta, \hat{\theta} \in \Theta$.

\medskip\textbf{Implementable-with-Payments Allocation Rule.}~
An allocation rule $f$ is implementable-with-payments if there exists a payment rule $p$ such that the mechanism $\mathcal{M} = \{f,p\}$ is incentive compatible.

\medskip\textbf{Deterministic Mechanism.}~
A mechanism is deterministic if each allocation $a \in A$ selects one assignment with probability $1$ and all the other outcomes with probability $0$, i.e., $a(s) = 1, s \in S$ and $a(s') = 0, \forall s' \neq s, s' \in S, \forall a \in A$.  In this sense, all mechanisms, including deterministic mechanisms, are randomized.

\medskip\textbf{Universally Truthful Mechanism.}~
We say an incentive compatible randomized mechanism $\mathcal{M} = (f,p)$ is univerally truthful if it is a distribution over incentive compatible deterministic mechanisms. That is, there is a set of deterministic incentive compatible mechanisms $S_{\mathcal{M}}$ known as the support of $\mathcal{M}$ and a probability distribution $\alpha$ such that  $f_{\mathcal{M}} = \sum_{\mathcal{M}' \in S_{\mathcal{M}}} \alpha_{\mathcal{M}'}  f_{\mathcal{M}'}$.

\medskip\textbf{Truthful in Expectation Mechanism.}~
If a randomized mechanism is incentive compatible, we say that it is truthful in expectation.  (The expected value is implicit in the dot product in the definition of $u^{\mathcal{M}}(\theta, \hat{\theta})$.) 

\medskip\textbf{Truthful with Verification Mechanism.}~
Let $V \subseteq \Theta \times \Theta$ be the set of pairs $(\theta,\hat{\theta})$ that the designer can verify and denote with $\mathcal{M}_v =(f,V)$ a mechanism with verification. $\mathcal{M}_v$ is truthful if for all pairs of types $(\theta,\hat{\theta})$ either $f(\theta) \cdot \theta \geq f(\hat{\theta}) \cdot \theta$ or $(\theta,\hat{\theta}) \in V$. Intuitively, either an agent with type $\theta$ prefers not to report $\hat{\theta}$ or the mechanism designer can detect or prevent such a report.

\section{Basics of Partial Verification}

In this section, we develop our basic tools for reasoning about what types do not need to be verified.  In particular, given a type $\theta$, if the agent can never benefit by reporting some other type $\hat{\theta}$, then from the perspective of the mechanism designer it is unnecessary to be able to verify $(\theta, \hat{\theta})$, i.e., verify that $\hat{\theta}$ is not the agent's true type.  We call such types {\em harmless}.~\footnote{
To simplify the characterization, we focus on harmless types instead of types that must be verified.  Given a type $\theta$ and allocation rule $f$, a pair $(\theta, \hat{\theta})$ is in $V$ if $\hat{\theta} \in \Theta \setminus H(\theta,f)$. For brevity, we say that the set of types which must be verified is $\Theta \setminus H(\theta,f)$.
That is, the mechanism will be truthful as long as the mechanism designer can verify that the agent does not have true type $\theta$ when she reported $\theta'$, for all $\theta' \in \Theta \setminus H(\theta,f)$.}
\begin{definition}
\label{def:harmless}
Given a type $\theta$ and allocation rule $f$, the \emph{harmless set of types} $H(\theta,f)$ is the set composed by the types $\hat{\theta} \in \Theta$ such that $f(\theta) \cdot \theta \geq f(\hat{\theta}) \cdot \theta$.
\end{definition}

We can also talk about the harmless set of types for multiple allocation rules.
This allows us to take our results about the single agent setting and apply them to settings with multiple agents.  
We can turn the multiagent allocation rule into a single agent allocation rule by ``plugging in'' the types of the other agents.  However, in doing so we end up with different allocation rules depending on the types of those other agents.  Thus, capturing a multiagent allocation rule in a single agent setting requires a whole set of allocation rules (see the Multi-Agent Mechanisms section).

As an added benefit, this generality allows us to talk about not just single mechanisms, but whole families of mechanisms.  Later, we exploit this flexibility to draw a sharp contrast between the harmless set for all universally truthful mechanisms and the harmless set for all truthful in expectation mechanisms.
\begin{definition}
Let $F$ be a set of allocation rules.  Then the \emph{harmless set of types} $H(\theta,F)$ is the set composed by the types $\hat{\theta} \in \Theta$ such that $f(\theta) \cdot \theta \geq f(\hat{\theta}) \cdot \theta$ for every $f \in F$.
\end{definition}

It is immediate from this definition that the harmless set of types of a set of allocation rules is the intersection of their individual harmless sets.
This is because our definition imposes a strong requirement for a type to be harmless: it identifies a type as harmless only if it is harmless for every allocation rule in the set.  That is, there is no scenario under which the agent can benefit from reporting a harmless type.  This strong definition is in the same spirit as the definition of incentive compatibility; we only want to declare a type harmless if the mechanism designer never has to worry about that type being reported.
\begin{observation}
\label{obs:intersection}
The harmless set of types $H(\theta, F)$ corresponds to the intersection of the harmless set of types of every allocation rule $f \in F$, i.e., $H(\theta, F) = \cap_{f \in F} H(\theta,f)$.
\end{observation}

Using $H(\theta, F)$, we can express the minimal verification needed to guarantee that a implementable-with-payments allocation rule is also truthful without them.  In particular, this minimal verification corresponds to the set $V$ where $(\theta, \hat{\theta}) \in V$ if and only if $\hat{\theta} \in \Theta \setminus H(\theta,f)$.

\begin{table}
  \begin{center} 
    \renewcommand{\arraystretch}{1.2}
    \begin{tabular}{@{}ll@{}}
      \toprule
      Notation & Definition
    \\
    \midrule
    $f$   &      Allocation rule \\ 
    $F$ &  Set of allocation rules\\ 
    $H(\theta,f)$ & \parbox[t]{0.6\linewidth}{Harmless set of types given true\\[-1pt] type $\theta$ and allocation rule $f$} \\[10pt] 
    $H(\theta,F)$ & \parbox[t]{0.6\linewidth}{Harmless set of types given true\\[-1pt] type $\theta$ and set of allocation rules $F$} \\[10pt] 
    $f_{\{a_i, a_j\} }$   &  Separating allocation rule \\ 
    $\bar{F}_{\{a_i, a_j\}}$    &    Set  of separating allocation rules $f_{\{a_i, a_j\} }$ \\ 
    $\bar{F}$ &  Set of $\bar{F}_{\{a_i, a_j\}}, \forall \{a_i, a_j\} \in A$  \\ 
    $l_{f, \{ a_i,a_j \}}$	 & Allocation hyperplane over $\{ a_i,a_j \}$ \\ 
    $I_{\{ a_i,a_j \}}$   & Indifference hyperplane \\ 
    $L_{\{ a_i,a_j \}}$   & Set of $l_{f, \{ a_i,a_j \}}$ parallel to $I_{\{ a_i,a_j \}}$ \\ 
    $f^{\theta}$      &  Critical allocation rule \\ 
    $l_{f^{\theta}, \{ a_i,a_j \}}$  & Critical allocation hyperplane \\ 
    $\bar{F}^{\theta}$ &  Set of critical allocation rules \\ 
    $L^{\theta}$    & Set of critical allocation hyperplanes \\ 
    \bottomrule
    \end{tabular}
\end{center}
\caption{Notation used throughout the paper.}
\label{tab:notation}
\end{table}

\subsection{Harmless sets with two allocations (Informally)}
\label{sec::2dim}
We now introduce our main tools for understanding the harmless sets of implementable-with-payments allocation rules.  At first, we characterize the harmless sets of implementable-with-payments allocation rules which have exactly two allocations in their range, and then show that there is a sense in which this captures everything we need to know about the harmless set of types even when there are more than two allocations.
Before giving a formal treatment of this setting, we walk through it more informally.

With only two allocations, incentive compatible mechanisms have a simple, well-known form.  By the taxation principle, any incentive compatible mechanism consists of assigning a price to each allocation and letting the agent choose which allocation it prefers to pay for~\cite{guesnerie1984}.  Thus, if we call the two allocations $a_1$ and $a_2$ and assign them prices $p_1 = p(a_1)$ and $p_2 = p(a_2)$, an agent with type $\theta$ can be assigned $f(\theta) = a_1$ by an incentive compatible mechanism only if $a_1 \cdot \theta - p_1 \geq a_2 \cdot \theta - p_2$ (and similarly for $a_2$).  Rewriting, it is easy to see that the types that are indifferent and could be assigned either allocation are those who satisfy $(a_1 - a_2) \cdot \theta = p_1 - p_2$.  That is, these types all lie on a hyperplane.  Further, the two half spaces on either side of this hyperplane correspond to the sets of types that prefer each allocation at the given prices,  i.e., if $\theta$ is in the interior of one halfspace and $\theta'$ is in the interior of the other then  $u^{\mathcal{M}}(\theta, \theta) > u^{\mathcal{M}}(\theta, \theta')$ and $u^{\mathcal{M}}(\theta', \theta') > u^{\mathcal{M}}(\theta', \theta)$.

Since such a hyperplane is uniquely identified by a {\em relative price} $c = p_1 - p_2$, every implementable-with-payments allocation rule $f$ can be associated with the hyperplane $(a_1 - a_2) \cdot \theta = c$ for some real number $c$.  Note however, that there will in general be many allocation rules associated with a single hyperplane because types on the hyperplane are indifferent between the allocations whose prices difference is $c$ and so can be assigned to either allocation by an implementable-with-payments allocation rule.

Now consider a particular such $f$ and a type $\theta$.  There are five possible cases for $H(\theta,f)$. 

 \textbf{Case 1:} $\theta \cdot a_1 > \theta \cdot a_2$ and $f(\theta) = a_1$. An agent with type $\theta$ is already receiving her preferred outcome, so the agent cannot gain by reporting another type.  Thus $H(\theta,f) = \Theta$.  

 \textbf{Case 2:} $\theta \cdot a_1 > \theta \cdot a_2$ and $f(\theta) = a_2$.  An agent with type $\theta$ can benefit by reporting any type $\theta'$ such that $f(\theta') = a_1$, so $H(\theta,f) = \Theta - \{ \theta' : f(\theta') = a_1$\}.  By the above analysis, $H(\theta,f)$ contains all types on the side of the hyperplane associated with $f$ where types prefer $a_2$ at relative price $c$ implied by $f$.  It may also contain some types on the hyperplane, if $f$ happens to assign them $a_2$. 

\textbf{Case 3:} $\theta \cdot a_1 = \theta \cdot a_2$. This case is degenerate and the agent with type $\theta$ is totally indifferent between the two allocations, so $H(\theta,f) = \Theta$ regardless of the $f$ chosen.

\textbf{Cases 4 and 5:} symmetric to Cases 1 and 2.

So what does $H(\theta,F)$ look like where $F$ is the set of all such $f$?  By Observation~\ref{obs:intersection}, we need to take the intersection of the harmless sets.  In the degenerate case 3, this yields $H(\theta,F) = \Theta$.  Otherwise, all that matters is the $f$ for which case 2 applies.  That is we care about the $f$ which correspond to hyperplanes with $c$ such that $(a_1 - a_2) \cdot \theta \leq c$.  The intersection of the harmless sets of all these hyperplanes is the set of $\theta'$ which are ``below'' all of them.  This is entirely determined by the ``lowest'' such hyperplane, the one where $(a_1 - a_2) \cdot \theta = c$. Consider the following example.

\begin{example}
\label{ex:2determistic}
Consider the case of a deterministic incentive compatible mechanism with two possible assignments, $s_1$ and $s_2$, and two allocations, $a_1$ and $a_2$, such that $a_1(s_1) =1$ and $a_2(s_2)= 1$, i.e., allocation $a_1$ assigns $s_1$ to the agent with probability $1$, while allocation $a_2$ assigns $s_2$ with probability $1$. Furthermore, assume that the agent's type is $\theta = (\theta_{s_1}, \theta_{s_2})$ with $\theta_{s_1} < \theta_{s_2}$.  This setting is illustrated in Figure~\ref{fig:det-mech-two-plus-zero} (a), where $\theta_1 = \theta_{s_1}$ and $\theta_2 = \theta_{s_2}$.

The hyperplane of types $\theta' \in \Theta$ for which $\theta'_{s_1} = \theta'_{s_2}$ is the 45 degree line from the origin, and which we refer to as the {\em indifference hyperplane}.  Note that it corresponds to taking $c = 0$, and that changing $c$ just translates this line while keeping it at 45 degrees. 
The translations of this line for which $(a_1 - a_2) \cdot \theta \geq c$, i.e., $\theta_{s_1} - \theta_{s_2} \geq c$, are the lines that in the figure would be above $\theta$; the lowest of these is the one which passes through $\theta$, which we refer to as the {\em critical allocation hyperplane}.  
The harmless set $H(\theta,F)$ is the set of types below this critical allocation hyperplane.  It corresponds to the types that prefer $s_1$ relative to $s_2$ more strongly that $\theta$.  That is, those $\theta'$ where $\theta_{s_1} - \theta_{s_2} < \theta_{s_1}' - \theta_{s_2}'$.
\end{example}

\subsection{Formal treatment of two allocations}
\label{sec::2dimformal}

We define the concepts introduced in the previous section and formally prove how to identify the harmless set of types of implementable-with-payments allocation rules.
We start by defining a separating allocation rule i.e., an allocation rule that can be associated with a hyperplane that divides the space in two half spaces such that all the types in one half space receive the same allocation, and a set of such allocations.

\begin{definition}
An allocation rule $f_{\{a_i, a_j\} }$ is \emph{separating} if\\ $f_{\{a_i, a_j\} }: \Theta \rightarrow \{a_i, a_j\} \subseteq A$ and there exists a hyperplane which separates the type space $\Theta$ in two open half-spaces $\Theta', \Theta'' \subseteq \Theta$ such that the closure of their union is $\Theta$ and if $\theta \in \Theta'$ then$f(\theta) = a_i$ while if $\theta \in \Theta''$ then $f(\theta) = a_j$.  (For brevity, when the allocation pair $\{a_i, a_j\}$ is clear from context we suppress it and simply write $f$.)\footnote{Note that any implementable-with-payment allocation rule is also a separating allocation rule.}
\end{definition}

\begin{definition}
Let $\bar{F}_{\{a_i, a_j\}}$ denote the set of separating allocation rules $f_{\{a_i, a_j\} }$.
Then let  $\bar{F} = \cup_{\{a_i, a_j\} \subseteq A}\bar{F}_{\{a_i, a_j\}}$.
\end{definition}

Given a separating allocation rule, we are interested in the hyperplane it induces, in the following sense.

\begin{definition}
The allocation hyperplane $l_{f,A'}$ over allocation set $A' = \{a_i, a_j\}$ is the hyperplane that separates the two half-spaces identified by the separating allocation rule $f \in \bar{F}_{\{a_i, a_j\}}$. In the remaining of the paper, we will say that $l_{f,A'}$ is {\em induced by} the allocation rule $f \in \bar{F}_{\{a_i, a_j\}}$.
\end{definition}

Of course, selecting two allocations and a hyperplane is not sufficient for an allocation rule to be implementable-with-payments.  By the taxation principle, the hyperplane must consist of all the types which are indifferent between the two allocations at a particular price.  Further, the remaining types must receive the ``correct'' allocation.  That is, those which would be willing to pay more than the price to get one allocation instead of the other are the ones that receive it.  Such hyperplanes are exactly those parallel to the hyperplane of types indifferent between the two allocations. 
We capture these requirements in the following definitions.

\begin{definition}
Given $\{a_i, a_j \} \subseteq A$, the \emph{indifference hyperplane} $I_{\{a_i, a_j\}}$ is the hyperplane
composed of types where the agent is indifferent between allocation $a_i $ and  allocation $a_j$, i.e. all the points $\theta' \in \Theta$ where $a_i \cdot \theta' = a_j \cdot \theta'$.
\end{definition}

\begin{definition}
Let $L_{\{a_i, a_j\}}$ be the set of allocation hyperplanes $l_{f,\{a_i, a_j\}}$ parallel to indifference hyperplane $I_{\{a_{i}, a_{j}\}}$.
\end{definition}

The following observation formally summarizes the preceding discussion by showing that the hyperplanes in the set $L_{\{a_i, a_j\}}$ are only the ones that are induced by an implementable-with-payments allocation rules given the allocations $\{a_i, a_j\}$, and thus that the implementable-with-payments allocation rules are separating allocation rules.

\begin{observation}
\label{obs:power}
A hyperplane is in $L_{\{a_i, a_j\}}$ if and only if it is induced by an implementable-with-payments allocation rule $f \in \bar{F}_{\{a_i, a_j \}}$.
\end{observation}

\begin{proof}
First, note that if $l_{f,\{a_i, a_j\}} = I_{\{a_{i'}, a_{j'}\}}$ and the payment is equal to zero, then the agent has no incentive to misreport her type, i.e., the mechanism $\mathcal{M} = (f, 0)$ is incentive compatible.
Second, we know that, from the taxation principle, an allocation rule $f \in \bar{F}_{\{a_i, a_j \}}$ is truthfully implementable-with-payments in dominant strategies if the agent is charged the same payment every time she receives a given assignment. 
Equivalently,~\footnote{By the previous sentence, we know that there exists a constant $c$ (equal to the payment difference)  such that if $a_i\cdot\theta - c > a_j\cdot\theta$ then the agent receives $a_i$ and vice versa.  This equation exactly corresponds to a translation of the indifference hyperplane by $c$.} the allocation hyperplane $l_{f,\{a_{i}, a_{j}\}}$ of an implementable-with-payments allocation rule $f \in \bar{F}_{\{a_i, a_j \}}$ is parallel to the related indifference hyperplane $I_{\{a_{i}, a_{j}\}}$.
\end{proof}

While there are many allocation hyperplanes in $L_{\{a_i,a_j\}}$, 
the harmless set is entirely determined by one of them, the one identified as the "lowest" in the previous section, which we term the {\em critical allocation hyperplane.}

\begin{definition}
The \emph{critical allocation hyperplane} $l_{f^{\theta},{\{a_i, a_j\}}} \in L_{\{a_i, a_j\}}$ is parallel to the indifference hyperplane $I_{\{a_i, a_j\}} $ and the agent's type belongs to it (i.e., $\theta \in l_{f^{\theta},{\{a_i, a_j\}}} $).
We call a rule $f^{\theta}\in \bar{F}_{\{a_i, a_j\}}$ that induces $l_{f^{\theta},{\{a_i, a_j\}}}$ a critical allocation rule.
\end{definition}

Note that there exist an infinite number of critical allocation rules that induce a critical allocation hyperplane. In the remaining of the paper, the critical allocation rule we consider is the following: if $\theta' \in l_{f^{\theta}, \{a_i,a_j\}}$ and $a_i \cdot \theta' > a_j \cdot \theta'$, then $f^{\theta}(\theta') = a_i$ but $f^{\theta}(\theta) = a_j$.  I.e., $\theta$ gets its less preferred allocation, while all the other types on the critical allocation hyperplane get the more preferred allocation.  This implies that $\theta' \notin H(\theta, \bar{F})$, $\forall \theta' \neq \theta \in l_{f^{\theta}, \{a_i,a_j\}}$.
\begin{lemma}
\label{lem:critical}
$H(\theta, \bar{F}_{\{a_i, a_j \}}) = H (\theta, f^{\theta})$, i.e., the harmless set of types of the set of rules $\bar{F}_{\{a_i, a_j \}}$ is equal to the harmless set of types of allocation rule $f^{\theta} \in \bar{F}_{\{a_i, a_j \}}$ that induces the critical allocation hyperplane $l_{f^{\theta},{\{a_i, a_j \}}} $.
\end{lemma}
\begin{proof}
First notice that, since the allocation hyperplanes in $L_{\{a_i, a_j\}}$ are parallel, the critical allocation hyperplane $l_{f^{\theta}, \{a_i,a_j\}}$ divides the allocation hyperplanes in $L_{\{a_i, a_j\}}$ in two sets, depending on which side of it they lie.  For those on the same side as the indifference hyperplane, the agent is already getting its preferred type so $H(\theta, f) = \Theta$.  For those on the opposite side, the agent would rather report a type yielding her preferred allocation, so $H(\theta, f)$ corresponds to the open half-space containing the indifference hyperplane.  The intersection of all these sets is exactly  $H(\theta, f^{\theta})$
\end{proof}
The previous lemma implies that, given $\theta$ and a set of allocations $A' = \{a_i, a_j \}$, the critical allocation hyperplane $l_{f^{\theta},A'} $ divides the space into two half-spaces, and that if $l_{f^{\theta},A'} \neq I_{A'} $, then the open half-space containing the indifference hyperplane corresponds to the harmless set of types of $\bar{F}_{\{a_i, a_j \}}$, otherwise, if $l_{f^{\theta},A'} = I_{A'} $, then $H(\theta, \bar{F}_{\{a_i, a_j \}}) = H(\theta, f^{\theta}) = \Theta$.
\begin{definition}
Let $\bar{F}^{\theta} = \{ f^{\theta}_{\{a_i, a_j \}} : \{a_i, a_j \} \in A\}$ be the set of critical allocation rules given all pairs $\{a_i, a_j \} \in A$.
\end{definition}
\begin{definition}
Let $L^{\theta} =\{ l_{f^{\theta},{\{a_i, a_j\}}} : \{a_i, a_j\} \in A \}$ denote the set of critical allocation hyperplanes.
\end{definition}

From Observation~\ref{obs:intersection} and Lemma~\ref{lem:critical}, follows Corollary~\ref{cor:final}, which says that to identify the harmless set it suffices to identify the critical allocation hyperplanes.

\begin{corollary}
\label{cor:final}
$H(\theta, \bar{F}) = H(\theta, \bar{F}^{\theta}) = \cap_{f^{\theta} \in \bar{F}^{\theta}} H(\theta, f^{\theta})$.
\end{corollary}

Figure~\ref{fig:det-mech-two-plus-zero}(a) shows an example of a critical allocation hyperplane, indifference hyperplane, and harmless set of types for a set of implementable-with-payments allocation rules for the case with two allocations.

This example also provides a clear geometric explanation for the phenomenon observed in previous work that ``symmetric'' verifications (which tend to take the form of a constraint that misreports must be local to the true type) do not tend to help while ``asymmetric'' ones do~\cite{fotakis2015b}.  Because $\theta$ is on the critical allocation hyperplane, there are arbitrarily close misreports which can lead to a benefit for some allocation rules, so restricting misreports to be close to the true type does not help.  In contrast, an asymmetric verification which rules out the entire half space above the critical allocation hyperplane is very useful.

\section{More Than Two Allocations}
\label{sec::more_than_two}

Now that we understand how to identify harmless sets of types of implementable-with-payments allocation rules with two allocations, we can extend our analysis to cases with more than two allocations.  The key observation is that if a type $\theta'$ is not harmless then there exists an allocation rule $f$ and choice of $a_1$ and $a_2$ such that $f(\theta') = a_1$ while $f(\theta) = a_2$ but $\theta \cdot a_1 > \theta \cdot a_2$.  Since only these two allocations are relevant, we can actually find an implementable-with-payments allocation rule  for which $\theta'$ is not harmless that only allocates $a_1$ and $a_2$, thus reducing to the case we have already analyzed.  To identify the harmless set when there are more than two allocations, we thus intersect the harmless sets resulting from each pair of allocations.

\begin{theorem}
\label{thm:multipleAllocations}
Let $F$ be the set of implementable-with-payments allocation rules.  $H(\theta, F) = H(\theta, \bar{F}^{\theta})$. 
\end{theorem}

\begin{proof}
Because $\bar{F}^{\theta} \subset F$, $H(\theta, F) \subset H(\theta, \bar{F}^{\theta})$. 
For the other direction, let $\theta'$ be given such that $\theta' \not\in H(\theta, F)$.
By Definition~\ref{def:harmless}, $\theta' \notin  H(\theta, F)$ if and only if there exists an allocation rule $f \in F$ such that $\theta \cdot f(\theta') > \theta \cdot f(\theta)$.
By the taxation principle, we can represent $f$ by a list of allocations and the price for each allocation.  Construct $f'$ from $f$ by eliminating all allocations except $f(\theta)$ and $f(\theta')$ from this list.  Then $f'$ is implementable-with-payments (those from the list), $f'(\theta) = f(\theta)$, $f'(\theta') = f(\theta')$, and $f' \in \bar{F}$.  Thus $\theta' \not\in H(\theta, \bar{F})$.  By Corollary~\ref{cor:final}, $\theta' \not\in H(\theta, \bar{F}^{\theta})$.
\end{proof}

\section{Verification and Randomization}

In this section, we examine the implications of allowing randomization for implementing mechanisms using partial verification.  We show that there is a significant harmless set shared by all deterministic mechanisms.  Since universally truthful mechanisms are just distributions over deterministic mechanisms, this turns out to be true for them as well.  
However, the harmless set shared by all truthful in expectation mechanisms is quite limited.

\subsection{Deterministic mechanisms}

We now study how to identify the harmless set of types for all truthful deterministic mechanisms.  The result naturally follows from Theorem~\ref{thm:multipleAllocations}. In particular, the harmless set is the intersection of the harmless sets of all deterministic mechanisms with two allocations, which in turn corresponds to the intersections of the harmless sets generated by the relevant critical allocation hyperplanes.
\begin{theorem}
\label{thm:deterministic}
Let $F$ be the set of deterministic implementable-with-payments allocation rules using allocations in $A$. Then, $H(\theta, F)$ corresponds to the intersection of the half-spaces generated by all $l_{f^{\theta}, \{a_i, a_j\}} \in L^{\theta}$ containing the origin.
\end{theorem}
\begin{proof}
First note that since the mechanisms here considered are deterministic there exists one allocation for each possible assignment, i.e., $A = \{a_1, a_2, \dots, a_{|S|} \}$. For every pair of allocations $\{a_i, a_j\} \in A$  the harmless set of types $H(\theta, \bar{F}_{\{a_i, a_j\} })$  can be computed as shown in Example~\ref{ex:2determistic}, i.e., $H(\theta, \bar{F}_{\{a_i, a_j \}}) = H(\theta, f^{\theta})$ where $H(\theta, f^{\theta})$ corresponds to the open half-spaces generated by $l_{f^{\theta}, \{a_i, a_j\}} \in L^{\theta}$ that contains the origin.
Thus, due to Corollary~\ref{cor:final}, to compute $H(\theta, \bar{F})$, we need to consider only the allocation rules $f^{\theta} \in \bar{F}^{\theta} $ where $\bar{F}^{\theta}$ is the set of the critical allocation hyperplanes, one for each $\{a_i, a_j\} \in A$.
Given this, Observation~\ref{obs:intersection}, and Theorem~\ref{thm:multipleAllocations}, we can conclude that $H(\theta, F) = H(\theta, \bar{F^{\theta}}) = \cap_{f^{\theta} \in F^{\theta}} H(\theta, f^{\theta})$.
\end{proof}

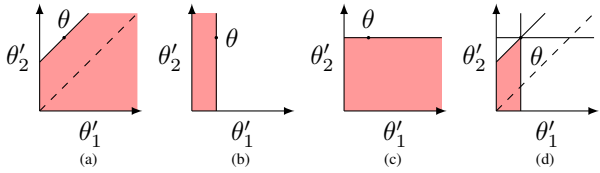
\begin{figure}[t]
  \begin{tikzpicture}[scale=0.65]
    \fill[harmless] (0,0) -- (0,1) -- (1,2) -- (2,2) -- (2,0) -- cycle;
    \axes
    \draw[indifference] (0,0) -- (2,2);
    \draw[critical] (0,1) -- (1,2);
    \putpoint{0.5,1.5}{above}{\theta}
    \node[] at (1, -1) {\tiny{(a)}};
  \end{tikzpicture}
  \begin{tikzpicture}[scale=0.65]
    \fill[harmless] (0,0) -- (0.5,0) -- (0.5,2) -- (0,2) -- cycle;
    \axes
    \draw[critical] (0.5,0) -- (0.5,2);
    \putpoint{0.5,1.5}{right}{\theta}
    \node[] at (1, -1) {\tiny{(b)}};
  \end{tikzpicture}
  \begin{tikzpicture}[scale=0.65]
    \fill[harmless] (0,0) -- (0,1.5) -- (2,1.5) -- (2,0) -- cycle;
    \axes
    \draw[critical] (2,1.5) -- (0,1.5);
    \putpoint{0.5,1.5}{above}{\theta}
    \node[] at (1, -1) {\tiny{(c)}};
  \end{tikzpicture}
%
  \begin{tikzpicture}[scale=0.65]
    \fill[harmless] (0,0) -- (0,1) -- (0.5,1.5) -- (0.5,0) -- cycle;
    \axes
    \draw[indifference] (0,0) -- (2,2);
    \draw[critical] (0,1) -- (1,2);
    \draw[critical] (0.5,0) -- (0.5,2);
    \draw[critical] (2,1.5) -- (0,1.5);
    \putpoint{0.5,1.5}{below right}{\theta}
    \node[] at (1, -1) {\tiny{(d)}};
  \end{tikzpicture}
\caption{Harmless sets for individual pairs of allocations for types with $\theta_\emptyset' = 0$ and their intersection (right graph).}
\label{fig:det-mech-two-plus-zero}
\end{figure}
We provide an example, shown in Figure~\ref{fig:det-mech-two-plus-zero}, to illustrate how to compute the harmless set of types for deterministic mechanisms with more than two allocations.
\begin{example}
\label{example:deterministic}
Consider a case with three assignments, one of which is null with no value.  Thus, we have assignments $S = \{\emptyset, 1, 2\}$ and allocations $A = \{a_1, a_2, a_3 \}$. Without loss of generality, assume that $a_1(\emptyset) = 1$, $a_2(1) = 1$, and $a_3(2)= 1$. Consequently, $\bar{F}^{\theta} = \{f^{\theta}_{\{a_1, a_2\}}, f^{\theta}_{\{a_1, a_3\}}, f^{\theta}_{\{a_2, a_3\}}\}$. The harmless sets for these allocation rules are shown in Figures~\ref{fig:det-mech-two-plus-zero} (b), (c), and (a) respectively.  Since $H(\theta, \bar{F^{\theta}})$, given by their intersection, is shown in Figure~\ref{fig:det-mech-two-plus-zero}(d).
\end{example}
This example also illustrates a key point about implementing the verifications found by our method: despite the infinite type space and infinite set of allocation rules, we can express the properties we need to verify in terms of a finite number of halfspace constraints, which gives reason to believe they may be verifiable in practical situations.

\subsection{Universally truthful mechanisms}

As previously discussed, the harmless set for all universally truthful mechanisms is the same as for all deterministic mechanisms.  We observe that every universally truthful mechanism can be represented as a distribution over truthful deterministic mechanisms, and every deterministic mechanism is a universally truthful mechanism that chooses the specific deterministic mechanism with probability 1.

\begin{theorem}
\label{thm:universally_truthful}
The harmless set of types $H(\theta, F)$ of all single agent universally truthful mechanisms 
 is equal to the harmless set of types of all single agent truthful deterministic mechanism. 
\end{theorem}
\begin{proof}
By definition, in order to guarantee that a mechanism $\mathcal{M}$ is universally truthful, we need to guarantee that every deterministic mechanism $\mathcal{M'} \in S_{\mathcal{M}} $  is incentive compatible. This implies that the harmless set of a universally truthful mechanism is equal to the intersection of the harmless set of all the deterministic mechanisms in its support. 

Now note that every deterministic mechanism $\mathcal{M}'$ is a universally truthful mechanism that randomizes over $\mathcal{M}'$ with probability equal to $1$. 
Thus, in order to compute the harmless set of types of all mechanism $\mathcal{M}$, we need to compute the intersection of the harmless sets of types of all the deterministic mechanisms.
\end{proof}

\subsection{A useful lemma}

Before turning to our characterization for truthful in expectation mechanisms, we prove a more general lemma that considers cases with various sets of allocations.  It is applicable, for example, when modeling restrictions such as the existence of a null allocation for which agents are known to have value 0 (common in auction settings).

The lemma works in the case where the set of possible allocations is rich enough that differences between possible allocations form a linear space.  In particular, we cannot have a finite set of allocations, as we do in the deterministic setting.  The lemma states that in this rich setting, we can characterize the harmless set of types as those which can be expressed as a scaling down of the true type, plus some vector which is perpendicular to the space of allocation differences.

The intuition for the lemma is given in Figure~\ref{fig:TIE}.  The first few plots illustrate why, when the vector space of allocation differences is the entire space, the harmless set of types are only those which are scaled-down versions of the true type: for all others we can find a critical allocation hyperplane for which they are on the wrong side. Specifically, it is always possible to identify a hyperplane that contains $\theta$ for which the origin belongs to one half space and the type $\theta'$ belongs to the other half space, i.e., $\theta'$ is not harmless, unless $\theta'$ belongs to the segment that connects the origin to $\theta$.

When this vector space becomes smaller, however, we add ``unenforceable'' directions in the type space, as types that differ only by a vector perpendicular to all allocation differences cannot be separated by a critical allocation hyperplane.  These unenforceable directions in turn expand the set of harmless types.  As an example, consider the a setting with two assignments and an agent who is indifferent between them.  Rather than the harmless set being just types between the origin and $\theta$, it is actually the entire space because the agent is indifferent among all possible allocations.  More generally, this phenomenon occurs any time an agent's value for two assignments is tied, even if there are others over which she has a strict preference.

\begin{lemma}
  \label{lem:VerificationLocal-1}
  Let $A\subseteq \reals^m$ be a set of allocations and let type $\theta$ be given.  Define $A^2_\theta = \{(a,a') : a,a'\in A,\, (a-a')\cdot\theta \neq 0\}$ to be the pairs of allocations that $\theta$ is not indifferent between, and $D_\theta = \{\lambda(a-a') : (a,a') \in A^2_\theta,\, \lambda\in\reals\}$ to be the set of scaled differences between such pairs of allocations.  If $D_\theta$ is a linear subspace, then $H(\theta,F) = \{\lambda \theta + d^\perp : \lambda \leq 1, d^\perp \in D_\theta^\perp\} \cap \Theta$, where $D_\theta^\perp = \{v \in \reals^m : \forall d\in D_\theta,\, d\cdot v = 0\}$.  That is, the harmless types are those which are ``smaller than'' $\theta$, modulo directions not captured by $D_\theta$. (Proof in Appendix).
\end{lemma}

\newcommand{\dthp}{d_{\theta'}}
\newcommand{\dth}{d_{\theta}}
\begin{proof}
  First, we show sufficiency.  Let $\theta' = \lambda \theta + d^\perp$ for $\lambda \leq 1$, and let $d = \alpha(a' - a) \in D_\theta$ for $\alpha > 0$.  Without loss of generality, let $\theta$ prefer $a'$ to $a$, so that $\theta\cdot d > 0$.  Now suppose for contradiction that $f$ is some allocation rule such that $f(\theta) = a$ but $f(\theta') = a'$, and consider the allocation hyperplane $\ell$ between the $a$ cell and $a'$ cell for $f$.  By Observation~\ref{obs:power}, $d$ is normal to this boundary, oriented toward the $a'$ cell.  For $\theta$ to be in the $a$ cell and $\theta'$ in the $a'$ cell, we must therefore have $(\theta'-\theta)\cdot d > 0$.  But $\theta'-\theta = (1-\lambda)\theta + d^\perp$, and thus $(\theta'-\theta)\cdot d = (1-\lambda)\theta\cdot d + d^\perp\cdot d$.  Now note that by assumption $\theta\cdot d < 0$ and $1-\lambda \geq 0$, and by definition $d^\perp \cdot d = 0$, so in fact $(\theta'-\theta)\cdot d \leq 0$, which is a contradiction.  Thus every such $\theta'$ is harmless.

  For necessity, first suppose $\theta' = \lambda \theta + \dthp^\perp$ for $\lambda > 1$.  By definition of $D_\theta$, we have some $d = \alpha(a'-a) \in D_\theta$ for which $\theta\cdot d\neq 0$; without loss of generality we take $\theta\cdot d>0$.  Now consider the allocation rule $f(\theta'') = a$ if $\theta''\cdot d \leq \theta\cdot d$ and $a'$ otherwise.  (One can check that $f$ is implementable-with-payments, as the cell boundary between $a$ and $a'$ is perpendicular to $d$ by construction, and correctly oriented.\footnote{While this construction may appear to leverage tie breaking in $f$, it would hold just as well defining $f(\theta'') = a$ if $\theta''\cdot d < z(1+\lambda)/2$ and $a'$ otherwise.  Similarly, the remaining arguments in the proof need not depend on how allocation rules break ties.})  Now we have $f(\theta) = a$ by construction, and $f(\theta)\cdot\theta = a\cdot\theta < a'\cdot \theta = f(\theta')\cdot\theta$, implying that $\theta'$ is not harmless.

  Now project $\theta$ and $\theta'$ onto $D_\theta$ by writing $\theta = \dth + \dth^\perp$ and $\theta' = \dthp + \dthp^\perp$ for $\dth,\dthp \in D_\theta$, $\dth^\perp,\dthp^\perp \in D_\theta^\perp$.
  From the above two cases, we may assume that $\dthp \neq \lambda\dth$ for any $\lambda\in\reals$.  Thus, we may further project $\dthp$ onto $\dth$, writing $\dthp = \alpha\dth + \bar d$ for a nonzero orthogonal direction $\bar d$: $\bar d\neq 0$, $\bar d \cdot \dth = 0$, and $\alpha\in\reals$.  Finally, take $d = \bar d + \epsilon\dth$ for some $\epsilon>0$ sufficiently small (to be determined).

  As $D_\theta$ is a vector space, $d\in D_\theta$, so we may write $d = \lambda (a'-a)$ for some $\lambda>0$ and some pair of allocations $(a,a')\in A^2_\theta$ for which $\theta$ is not indifferent.  Indeed, we have $d\cdot\theta = (\dth^\perp + \bar d + \epsilon\dth)\cdot\dth = \epsilon\|\dth\|^2 > 0$, so $\theta$ prefers $a'$ to $a$.  Now let $f$ be the allocation function defined by $f(\theta'') = a$ if $d\cdot\theta'' \leq d\cdot\theta$ and $a'$ otherwise.  Clearly $f(\theta) = a$, so to show that $\theta'$ is not harmless, it suffices to show $f(\theta') = a'$.  This follows from a simple calculation: $d\cdot\theta' = \bar d\cdot\theta' + \epsilon\dth\cdot\theta' = \bar d \cdot (\alpha\dth + \bar d \cdot \dthp^\perp) + \epsilon \dth\cdot\theta' = \|\bar d\|^2 + \epsilon \dth\cdot\theta' > 0$ for sufficiently small $\epsilon$.
\end{proof}

\subsection{Truthful in expectation mechanisms}

 \begin{figure}
 \hspace{-0.5cm}
   \centering
   \begin{tikzpicture}[scale=0.6]
    \fill[harmless] ($(1,1.6)+0.8*(-0.8,0.5)$) -- ($(1,1.6)-1.2*(-0.8,0.5)$) -- (2,0) -- (0,0) -- (0,2) -- cycle;
    \axes
    \draw[dotted,thick] (0,0) -- (1.25,2);
    \draw[critical] ($(1,1.6)+0.8*(-0.8,0.5)$) -- ($(1,1.6)-1.2*(-0.8,0.5)$);
    \putpoint{1,1.6}{below}{\theta}
    \putpoint{1.125,1.8}{right}{\theta'}
    \node[] at (1, -1) {\tiny{(a)}};
  \end{tikzpicture}
  \begin{tikzpicture}[scale=0.6]
    \fill[harmless] (0,0) -- (0.2,0) -- (1.2,2) -- (0,2) -- cycle;
    \axes
    \draw[dotted,thick] (0,0) -- (1.25,2);
    \draw[critical] (0.2,0) -- (1.2,2);
    \putpoint{1,1.6}{right}{\theta}
    \putpoint{0.7,0.7}{right}{\theta'}
    \node[] at (1, -1) {\tiny{(b)}};
  \end{tikzpicture}
  \begin{tikzpicture}[scale=0.6]
    \fill[harmless] (0,0) -- (0,0.2) -- (1.3,2) -- (2,2) -- (2,0) -- cycle;
    \axes
    \draw[dotted,thick] (0,0) -- (1.25,2);
    \draw[critical] (0,0.2) -- (1.3,2);
    \putpoint{1,1.6}{right}{\theta}
    \putpoint{0.3,0.8}{above}{\theta'}
    \node[] at (1, -1) {\tiny{(c)}};
  \end{tikzpicture}
  \centering
  \begin{tikzpicture}[tdplot_main_coords,scale=1]
    \coordinate (orig) at (0,0,0);
    \draw[black!50] (orig) -- (1,0,0);
    \draw[black!50] (orig) -- (0,1,0);
    \draw[black!50] (orig) -- (0,0,1);
    \path (1,0,0) edge[->,>=latex] node[left] {$\theta_1'$} (1.5,0,0);
    \path (0,1,0) edge[->,>=latex] node[below] {$\theta_2'$} (0,1.5,0);
    \path (0,0,1) edge[->,>=latex] node[right] {$\theta_3'$} (0,0,1.5);
    \fill[black,fill opacity=0.2,draw=black] (1,0,0) -- (0,1,0) -- (0,0,1) -- cycle;
    \putpoint{0.5,0.7,1}{above right}{\theta}
    \fill[harmless] (orig) -- ($(0,0.2,0.5)$) -- ($1.5*(1,1,1)+(0,0.2,0.5)$) -- ($3*(0.5,0.3,0)$) -- cycle;
    \fill[black,fill opacity=0.2] (orig) -- ($(0,0.2,0.5)$) -- (0.1,0.3,0.6) -- ($1/0.8*(0.5,0.3,0)$) -- cycle;
    \draw[critical] ($(0,0.2,0.5)$) -- ($1.5*(1,1,1)+(0,0.2,0.5)$);
    \draw[dotted,thick] (0.5,0.7,1) -- (orig);
  \end{tikzpicture}
  \caption{Harmless sets for randomized allocation rules.  The dotted line depicts the scaled versions of $\theta$, $\{\lambda\theta:\lambda\in\reals\}$.}
  \label{fig:TIE}
\end{figure}
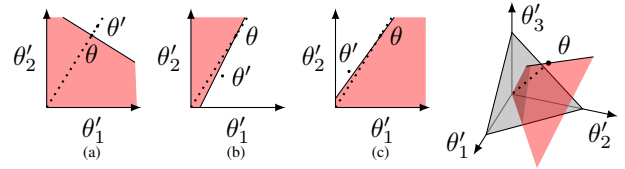

Our results for deterministic and universally truthful mechanisms are relatively positive, in that there is a significant harmless set of types which do not require verification.  For truthful in expectation mechanisms however, our results are much more negative.  Essentially, the only types in the harmless set are those which are a scaling or translation by adding the same constant to the value of each assignment of the original type, except in the special case where the agent is indifferent among two or more assignments, which adds additional dimensions to the harmless set.  For brevity, we state the theorem for the case where no such indifferences exist.

\begin{theorem}
\label{thm:truthful-expectation-no-zero}
Let $\theta$ be such that $\theta_{s_i} \neq \theta_{s_j}$ for all $i$ and $j$ and $m \geq 3$.
The harmless set of types of all single agent truthful in expectation mechanisms is $\{ \lambda \theta + \lambda'  \ones : \lambda \leq 1, \lambda' \in \reals \}$.
\end{theorem}

\begin{proof}
Theorem~\ref{thm:truthful-expectation-no-zero} is merely a special case of Lemma~\ref{lem:VerificationLocal-1} with difference set $D_\theta = \{d \in \reals^m : d\cdot \ones = 0\}$.  Consider the possible allocation differences in the setting of Theorem~\ref{thm:truthful-expectation-no-zero}.  We can write this set as $\{\alpha (p-q) : p,q \in \Delta_m, \alpha \in \reals\}$, which is all scaled differences of distributons over outcomes.  But this set is simply the affine hull of the probability simplex, $\{d\in\reals^m : d\cdot \ones = 1\}$, shifted so that it passes through the origin, which is precisely the vector subspace $D_\theta$ above.  Applying Lemma~\ref{lem:VerificationLocal-1} immediately gives the result, as here $D_\theta^\perp = \{0\}$.  Finally, note that adding indifferences reduces the dimension of $D_\theta$ and correspondingly increases the dimension of $D_\theta^\perp$, and thus the harmless set of types.
\end{proof}

The proof follows from a Lemma~\ref{lem:VerificationLocal-1} that encompasses the case with indifferences as well as more general scenarios.  
A direct intuition is shown in Figure~\ref{fig:TIE}.  Parts (b) and (c) show that we can find pairs of allocations where the critical allocation hyperplane is arbitrarily close to the line between $\theta$ and the origin. So, by considering the intersection of the harmless types of all the possible pair of allocations, the resulting harmless types must all be along this line.  Part (a) shows that types along the line from $\theta$ going away from the origin are not harmless. When these figures are combined in the full three dimensional type space, we gain an extra degree of freedom as we can add a constant to the value for each allocation without changing the incentives, resulting in the harmless set illustrated in part (d).

\section{Multi-Agent Mechanisms}
\label{sec:multi-agent-mech}

As further motivation for our characterizations of the implementability of all deterministic single-agent mechanisms in Theorem~\ref{thm:deterministic},  in this section, we discuss how to identify an agent's harmless set of types in multi-agent scenarios by leveraging on the results presented in the previous sections. Essentially, this boils down to a three step process:
\begin{enumerate}
\item Choose a truthful multi-agent mechanism $\mathcal{M}_{ma}$.
\item Derive a set of corresponding single-agent allocation rules $F_{sa}$ which are implementable with payments.
\item Characterize $H(\theta,F_{sa})$ for each single agent.
\end{enumerate}
To illustrate this process, consider a scenario with $n$ unit-demand agents and two  items, $i_1$ and $i_2$. Given this, the set of possible assignments for each agent is $S = \{\emptyset, i_1 , i_2 \}$.
Assume that the mechanism $\mathcal{M}_{ma}$ is the incentive compatible Vickrey-Clarke-Groves (VCG) auction~\cite{Vickrey1961,Clarke1971,Groves1973} that allocates items to agents such that the social welfare is maximised and charges each agent her externality. Thus, the allocation and the payment of each agent depends on the types reported by the other agents.
The next step is to derive the corresponding single-agent allocation rules $F_{sa}$.
In the case of VCG, every single-agent implementable-with-payments allocation rule $f_{sa}$ is characterized by prices $p_1$ and $p_2$ for item $i_1$ and $i_2$, respectively, which correspond to the agent's externality.  $f_{sa}$ then assigns the agent the item (or nothing) she prefers at those prices.
One of the $f_{sa}$ characterized by prices $p_1$ and $p_s$ is shown in Figure~\ref{fig:reserveprices}(a) where it is possible to observe that if the agent's type is in the red area no item is allocated to the agent, if it is in the blue area then she gets item $i_1$, and if it is in the green area she gets item $i_2$.  Without restrictions on the types of the other agents, every non-negative pair of prices $p_1$ and $p_2$ is possible, and thus this defines the set of single-agent allocation rules $F_{sa}$. Because every pair of prices is possible, we can immediately apply Theorem~\ref{thm:deterministic} for the final step to characterize the harmless set.\footnote{Strictly speaking VCG is a family of mechanisms determined by tie-breaking rules, our results apply to identify the set of types that is simultaneously harmless for all tie-breaking rules.}

\begin{figure}[t]
\hspace{-0.3cm}
  \begin{tikzpicture}[scale=0.65] 
    \fill[noitem] (0,0) -- (0,1.5) -- (0.5,1.5) -- (0.5,0) -- cycle;
    \fill[item2] (0,1.5) -- (0,2) -- (1,2) -- (0.5,1.5) -- cycle;
    \fill[item1] (0.5,0) -- (0.5,1.5) -- (1,2) -- (2,2) -- (2,0) -- cycle;
   \putpoint{0,1.9}{left} {\theta_2'};
   \path (0,0) edge[->,>=latex] node[below] {$\theta_1'$} (2.1,0);
    \draw[critical][->,>=latex] (0,0) -- (0,2);
    \draw[critical] (0.5,1.5) -- (1,2);
    \draw[critical] (0.5,0) -- (0.5,1.5);
    \draw[critical] (0.5,1.5) -- (0,1.5);
    \putpoint{0.5,0}{below}{p_1}
    \putpoint{0,1.5}{left}{p_2}
    \node[] at (1, -1) {\tiny{(a)}};
  \end{tikzpicture}
  \begin{tikzpicture}[scale=0.65]
    \fill[noitem] (0,0) -- (0,1.5) -- (1.5,1.5) -- (1.5,0) -- cycle;
   \path (0,0) edge[->,>=latex] node[below] {$\theta_1'$} (2.1,0);
    \draw[critical][->,>=latex] (0,0) -- (0,2);
    \putpoint{0,1.9}{left} {\theta_2'}; 
    \draw[critical] (1.5,1.5) -- (2,2);
    \draw[critical] (1.5,0) -- (1.5,1.5);
    \draw[critical] (1.5,1.5) -- (0,1.5);
    \putpoint{1.5,0}{below}{r_1}
    \putpoint{0,1.5}{left}{r_2}
	\putpoint{0.5,1}{left}{\theta}
	\node[] at (1, -1) {\tiny{(b)}};
	\fill[noitem] (3,0) -- (3,0.75) -- (4.5,0.75) -- (4.5,0) -- cycle;
	\fill[noitem] (3.25,0.75) -- (4.5,2) -- (4.5,0.75) -- cycle;
    \path (3,0) edge[->,>=latex] node[below] {$\theta_1'$} (5.1,0);
    \draw[critical][->,>=latex] (3,0) -- (3,2); 
    \putpoint{3,1.9}{left} {\theta_2'}; 
    \draw[critical] (3.25,0.75) -- (4.5,2);
    \draw[critical] (4.5,0) -- (4.5,2);
    \draw[critical] (4.5,0.75) -- (3,0.75);
    \putpoint{4.5,0}{below}{r_1}
    \putpoint{3,0.75}{left}{r_2}
	\putpoint{3.5,1}{left}{\theta}
	\node[] at (4, -1) {\tiny{(c)}};
	\fill[noitem] (6,0) -- (6,1.5) -- (6.25,1.5) -- (6.25,0) -- cycle;
	\fill[noitem] (6.25,0) -- (6.25,1.5) -- (6.5,1.75) -- (6.5,0) -- cycle;
    \path (6,0) edge[->,>=latex] node[below] {$\theta_1'$} (8.1,0);
    \draw[critical][->,>=latex] (6,0) -- (6,2); 
    \putpoint{6,1.9}{left} {\theta_2'}; 
    \draw[critical] (6.25,1.5) -- (6.75,2);
    \draw[critical] (6.25,0) -- (6.25,1.5);
    \draw[critical] (6.25,1.5) -- (6,1.5);
    \draw[critical] (6.5,0) -- (6.5,1.75);
    \putpoint{6.25,0}{below}{r_1}
    \putpoint{6,1.5}{left}{r_2}
	\putpoint{6.5,1}{right}{\theta}
	\node[] at (7, -1) {\tiny{(d)}};
  \end{tikzpicture}
\caption{(a) Example of  allocation rule $f_{sa}$ of single-agent mechanisms induced by a multi-agent mechanism where no item is allocated to the agent whose type is in the red area, if it is in the blue area then she gets item $i_1$, and if it is in the green area she gets item $i_2$. (b, c, d) Examples of how the Harmless set of types changes for different values of reserve prices $r_1$ and $r_2$ set by the mechanism.}
\label{fig:reserveprices}
\end{figure}
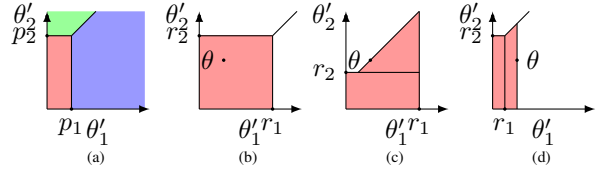

\begin{corollary}
\label{cor:VCG}
Let $F$ denote the set of implementable-with-payments deterministic allocation rules and let $F_{sa}$ denote the set of single-agent allocation rules derived from VCG. Then $H(\theta, F) = H(\theta,F_{sa})$.
\end{corollary}

This three step process can be applied to {\em any} truthful multi-agent mechanism.
In general, step 2 is an application of the taxation principle, and step 3 follows the logic of the proof of Theorem~\ref{thm:deterministic}.
Some cases, such as affine maximizers with finite agent weights and zero allocation weights~\cite{roberts79},
yield the same result as VCG, but others are more complex.  For example, in the same setting as before but with additional reserve prices $r_1$ and $r_2$, not every set of prices is possible, because $r_1$ and $r_2$ serve as lower bounds.  Thus, the harmless set of types depends also on the specific value of  $r_1$ and $r_2$ as shown in Figure~\ref{fig:reserveprices} (b,c,d).

\section{Allocation-dependent verification}

We have largely assumed that the set of verifiable types depends only on the true type.  Some authors assume, however, that the set of verifiable types also depends on the allocation received.  For example, in the combinatorial auction setting studied by  Fotakis, Krysta and Ventre~\citeyear{fotakis2014}, they assume that the mechanism designer  can only determine ex post whether the agent over-reported her value for the assignment she received. 

Our techniques still provide insight in this more refined setting.  Figure~\ref{fig:allocation-dependent} (a) shows both the harmless set (red) and the additional types not verified by their verification (blue) for the allocation rule that always assigns the agent her preferred assignment.  They are able to verify fewer types and still implement the allocation rule because, for this particular allocation rule the agent does not benefit while the harmless set satisfies the stronger condition that there is no rule under which they would benefit.  In fact our framework shows an even weaker verification would suffice, as any type below the horizontal line is harmless (see Figure~\ref{fig:det-mech-two-plus-zero} (c)).

\section{Reverse approach}

Our tools can be used also to answer the following question: given a reported type $\theta' \in \Theta$ and a class of mechanisms, what types need to be verified? In this case, the verification aims to check if a type $\hat{\theta} \in \Theta$ is the true type of the agent. Thus, from the perspective of the mechanism designer it is unnecessary to verify whether $\hat{\theta}$ is the agent's true type, if an agent with true type $\hat{\theta}$ cannot benefit by reporting $\theta'$. We call the types that need to be verified \emph{harmful}.

\begin{definition}
Given a reported type $\theta'$ and an allocation rule  $f$, the \emph{harmful set of types} $Z(\theta', f)$ is the set composed by the types $\hat{\theta} \in \Theta$ such that $f(\theta') \cdot \hat{\theta} > f(\hat{\theta}) \cdot \hat{\theta} $.
\end{definition}

Now, we show how to straightforwardly adapt our formulation to harmful sets of types.
Let's consider a particular such $f$ over two allocations with $c = p_1-p_2$ and a type $\theta'$. There are four possible cases for $Z(\theta', f)$.
 
\textbf{Case 1:} $\theta' \cdot a_1 \geq \theta' \cdot a_2$ and $f(\theta') = a_1$. An agent with any true type $\hat{\theta} \in \Theta$ such that $\hat{\theta} \cdot a_1 > \hat{\theta} \cdot a_2$ and $f(\hat{\theta}) = a_2$ can benefit by reporting $\theta'$. Thus, such types are the ones in the set $Z(\theta', f)$ and corresponds to all and only the types that belong to both the same half space as $\hat{\theta}$ w.r.t. the indifferent hyperplane and the opposite half space as $\hat{\theta}$ w.r.t. the allocation hyperplane.
 
\textbf{Case 2:} $\theta' \cdot a_1 \geq \theta' \cdot a_2$ and $f(\theta') = a_2$. No agent can strictly benefit from reporting such a type $\theta'$ because either the agent prefers $a_1$, and thus would be made worse off by doing so, or the agent prefers $a_2$ and so must already have $f(\theta) = a_2$. Thus $Z(\theta', f) = \emptyset$.
 
\textbf{Case 3:} $\theta' \cdot a_1 \leq \theta' \cdot a_2$ and $f(\theta') = a_2$. As in Case 1.
 
\textbf{Case 4:} $\theta' \cdot a_1 \leq \theta' \cdot a_2$ and $f(\theta') = a_1$. As in Case 2.

We now move to consider a set of allocation rules $F$.
Interestingly, we can answer two slightly different versions of the previous question. One is to identify the set of types $Z(\theta', F)$ that are harmful for all the allocation rules in $F$, or the set of types $\bar{Z}(\theta', f)$ that are harmful for least one allocation rule in $F$. Thus, the set $Z(\theta', F)$ is equal to the intersection of $Z(\theta', f), \forall f \in F$, while $\bar{Z}(\theta', f)$ is equal to the union of  $Z(\theta', f), \forall f \in F$.
These sets are shown in Figure 4 (b) for Example~\ref{example:deterministic}.
In this case, the set $Z(\theta', F)$ corresponds to the types that belong to the red area, while $\bar{Z}(\theta', F)$ includes also the types in the blue areas.

These two formulations can be thought of as upper and lower bounding the verification needed respectively.
The $Z(\theta', F)$ formulation captures what needs to be verified using ordinary verification, which defines verifications solely in terms of $(\theta,\theta')$ pairs.  More refined versions of verification, such as allocation-dependent verification, can potentially verify fewer types by conditioning the verification on the specific allocation rule.  Thus, $Z(\theta', F)$ captures a stronger notion of the types for whom reporting $\theta'$ is strictly better in every scenario.  Such a lower bound could be used, for example, to argue that there is no practical verification in a particular scenario even if we include refined notions of verification.

As mentioned before, we have chosen to present our primary approach as identifying the harmless set of types given a set of agents, allocation rules, and agents' true types  because it leads to appealing geometrical characterization. The reverse approach is more natural for direct application by a mechanism designer because it is directly phrased in terms of what to do for a given report.  In particular, the steps the designer has to follow to use verification as a substitute for money are the following. First, the designer decides which family of implementable-with-payments allocation rules to use and collects the agents' reported types. Then he verifies that each agent's true type is not in the set of types $Z(\theta', F)$ and, if necessary, penalizes the agents by, e.g., excluding them in the allocation. Finally, he applies the chosen allocation rule.    The downside of the reverse approach is that the geometric characterization is more complex.  In the end however, the two are equivalent as all that matters is identifying the set of $(\theta,\theta')$ pairs for which verification is needed.

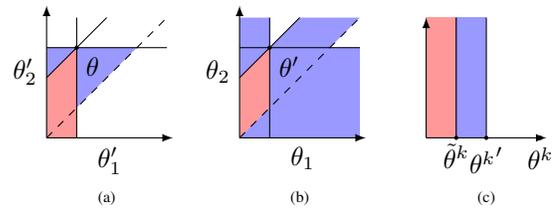
\begin{figure}[t]
  \centering
  \begin{tikzpicture}[scale=0.8]
    \fill[harmless] (0,0) -- (0,1) -- (0.5,1.5) -- (0.5,0) -- cycle;
    \fill[fotakis] (0,1) -- (0,1.5) -- (0.5,1.5) -- cycle;
    \fill[fotakis] (0.5,0.5) -- (0.5,1.5) -- (1.5,1.5) -- cycle;
    \axes
    \draw[indifference] (0,0) -- (2,2);
    \draw[critical] (0,1) -- (1,2);
    \draw[critical] (0.5,0) -- (0.5,2);
    \draw[critical] (2,1.5) -- (0,1.5);
    \putpoint{0.5,1.5}{below right}{\theta}
    \node[] at (1, -1) {\tiny{(a)}};
  \end{tikzpicture}~~~~
  $\qquad$
\hspace{-1cm}
  \begin{tikzpicture}[scale=0.8]
    \fill[fotakis] (2,0) -- (2,1.5) -- (1.50, 1.5) -- (2,2) -- (1,2) -- (0.5,1.5) -- (0.5,0.5) --(0,0) -- cycle;
     \fill[fotakis] (0,1) -- (0,2) -- (0.5,2) -- (0.5, 1.5) -- cycle;
	\fill[harmless] (0,0) -- (0,1) -- (0.5,1.5) -- (0.5,0.5) -- cycle;
    \axesaxes
    \draw[indifference] (0,0) -- (2,2);
    \draw[critical] (0,1) -- (1,2);
    \draw[critical] (0.5,0) -- (0.5,2);
    \draw[critical] (2,1.5) -- (0,1.5);
    \putpoint{0.5,1.5}{below right}{\theta'}
    \node[] at (1, -1) {\tiny{(b)}};
  \end{tikzpicture}
  \hspace{0.5cm}
    \begin{tikzpicture}[scale=0.8]
    \fill[harmless] (0,0) -- (0.5,0) -- (0.5,2) -- (0,2) -- cycle;
    \fill[fotakis] (0.5,0) -- (1,0) -- (1,2) -- (0.5,2) -- cycle;
    \path (0,0) edge[->,>=latex] node[below] {  } (2,0);
    \path (0,0) edge[->,>=latex] node[below] {  } (0,2);
    \putpoint{1.9,0}{below}{\theta^k}
    \draw[critical] (0.5,0) -- (0.5,2);
    \draw[critical] (1,0) -- (1,2);
    \putpoint{0.5,0}{below = -0.05}{\tilde{\theta}^{k}};
    \putpoint{1,0}{below}{\theta{^k}'};
    \node[] at (1, -1) {\tiny{(c)}};
  \end{tikzpicture}
\caption{(a) Example of the harmless set of types for a particular allocation when the verification is allocation-dependent. (b) Set of types $Z(\theta', F)$ (red area) and set of types $\bar{Z}(\theta', F)$ (red and blue areas) for Example~\ref{example:deterministic} (illustrated in Figure~\ref{fig:det-mech-two-plus-zero}).
(c) Single item auctions: Set of harmful types $Z_k({\theta^k}',f_{VCG}, \tilde{\theta}^k)$ (red area) for the case in which $\tilde{\theta}^k \leq {\theta^k}'$; set of harmful types $Z_k({\theta^k}',f_{VCG})$ (red and blue areas). }
\label{fig:allocation-dependent}
\end{figure}

\section{Application examples}
We conclude with three additional applications.
First we show an example of allocation dependent verification and reverse approach applied to second price auctions. For $k$-minded combinatorial auctions, we show that can recover previous results about when a particular verification is or is not sufficient and that we can extend them by characterizing a verification that would be sufficient for the case where it is not.  For $K$-facility location, we show how our framework allows us to recover a sufficient verification for a particular class of mechanisms and extend it to a much larger class.

\subsection{Second price auctions}
Consider the single item auction problem and the second price auction mechanism whose allocation rule, denoted by $f_{SP}$, assigns the single item $i$ to the agent $k \in K$ with the highest submitted bid ${\theta^k}'$.
From a single agent perspective, there are two possible assignments, namely $S = \{\emptyset, i \}$. Given this, a deterministic mechanism has two possible allocations, $a_1(\emptyset) = 1$ and $a_2(i) = 1$.
Here, we focus on the reverse approach with the aim to identify the set of types that needs to be verified for both the allocation-dependent case and non-allocation-dependent case (i.e., the one presented at the beginning of the paper).
 
In the allocation-dependent case, the designer takes into account $f_{SP}$ and observes the reported type ${\theta^k}'$ for all $k \in K$, while the real type $\theta^k$ of agent $k \in K$ is unknown.
Note that given any agent $k \in K$ and the allocation rule $f_{SP}$, only the highest bid of agents $K \setminus k$ affect $k$'s single agent mechanism; denote such bid with $\tilde{\theta}^k = \max \{{\theta^j}' | j \in K \setminus k\}$, and the single agent allocation rule that takes this information into account with $Z_k(\theta{^k}',f_{SP}^{\tilde{\theta}^k})$. Given this, we can focus on agent $k$'s single agent mechanism with threshold defined by $\hat{\theta}^k$ (as explained in the Multi-Agent Mechanisms section). In what follows, we focus on single agent mechanism and thus, for the sake of simplicity, we remove the index $k$. 
In particular, we denote with $Z(\theta',f_{SP}^{\tilde{\theta}})$ agent $k$'s set of harmful types given the allocation rule and the effect of the agents $K \setminus k$. 
The set $Z(\theta',f_{SP}^{\tilde{\theta}})$ for the case in which $\tilde{\theta} \leq \theta'$ corresponds to the red area in Figure~\ref{fig:allocation-dependent} (c). This set contains all the types $\hat{\theta} \in \Theta$ such that $f_{SP}^{\tilde{\theta}}(\hat{\theta}) \cdot \hat{\theta} < f_{SP}^{\tilde{\theta}}(\theta') \cdot \hat{\theta}$. Recall that for the allocation-dependent case, the verification happens only if the agent receives the item. Thus, $Z(\theta',f_{SP}^{\tilde{\theta}}) = \emptyset$ when $\tilde{\theta} > \theta'$.
 
In the non-allocation-dependent case the set $Z(\theta',f_{SP})$ is exactly as explained in the Reverse approach section and corresponds to the red and blue areas in Figure~\ref{fig:allocation-dependent} (c).
 
This example also highlights how easy is to compare different types of verification using our approach. Indeed, in this case, Figure~\ref{fig:allocation-dependent} (c) clearly illustrates how $Z(\theta',f_{SP}) \supseteq Z(\theta',f_{SP}^{\tilde{\theta}})$, and the scenarios in which this improvement is significant.

\subsection{$k$-minded combinatorial auctions}
\label{sec:appl-k-mind}

Consider the (known) $k$-minded combinatorial auction setting studied by Fotakis, Krysta and Ventre~\citeyear{fotakis2014}.  In this setting a set of goods must be allocated to a set of agents, and an agent has some value for exactly $k$ subsets of them.  (More precisely, she receives some set of items and her utility is that of the most valuable of the $k$ sets of which they are a superset.)  They showed that for $k = 1$, all implementable-with-payments allocation rules are also implementable using a verification that prevents agents from overbidding, while for $k > 2$ this is not the case.  This result follows easily from our results, that also provide a nice visual intuition for what goes wrong in the $k = 2$ case.

For $k = 1$, from a single agent perspective there are effectively two possible assignments, $S = \{s_1,s_2\}$: the agent does not get her desired bundle $\theta_{s_1} = 0$ or she does and gets value $\theta_{s_2}$.  From Theorem~\ref{thm:multipleAllocations} (deterministic mechanisms) or Theorem~\ref{thm:truthful-expectation-no-zero} (randomized mechanisms), we see that the harmless types are exactly those where the agents underbids.  Thus, being able to verify the agent did not overbid suffices.

For $k=2$, we simply add a new assignment $s_3$.  Letting $\theta_2 = \theta_{s_3}$ and $\theta_1 = \theta_{s_2}$, Figure~\ref{fig:det-mech-two-plus-zero}(a) shows the harmless set for deterministic mechanisms with the $s_1$ dimension omitted as $\theta_{s_1} = 0$.  Again applying Theorem~\ref{thm:multipleAllocations}, the harmless set no longer includes all types where the agent underbids.  In the example shown, the agent prefers $s_3$ to $s_2$, and so types where the agent underbids on $s_3$ but underbids more on $s_2$ are not harmless. Thus, this is exactly the sort of misreport that makes being able to verify that the agent has not overbid insufficient.  It also shows that a sufficient verification is that the agent has correctly reported her value for her preferred assignment and not overreported her value for the other assignment. (This can even be weakened to not overbidding on the assignment received, and, if the agent receives her preferred assignment, that she additionally did not underbid her value relative to the other assignment.)
Whether this verification is reasonable or not depends on the application.

For randomized mechanisms, the technical Lemma~\ref{lem:VerificationLocal-1} used to prove Theorem~\ref{thm:truthful-expectation-no-zero} can be directly applied to yield the following theorem.
\begin{theorem}
\label{thm:truthful-expectation-with-zero}
Let $\theta$ be such that $\theta_{s_i} \neq \theta_{s_j}$ for all $i$ and $j$ and $m \geq 3$.
The harmless set  of all single agent truthful in expectation mechanisms with $\theta_{s_1} = 0$ is $\{ \lambda \theta: \lambda \leq 1\}$.
\end{theorem}
\begin{proof}
Apply Lemma~\ref{lem:VerificationLocal-1}
with the expanded $D_\theta = \reals^m$ (because the total probability of non-zero assignments can now be less than one).  Thus $D_\theta^\perp = \{0\}$ and the result follows.
\end{proof}
Note that, in contrast to Theorem~\ref{thm:truthful-expectation-no-zero}, adding a constant to the value of each non-null allocation is no longer harmless because it changes values relative to the null allocation.

\subsection{$K$-facility location}

Consider a set of $G$ potential locations where a set of $K$ facilities will ultimately be located ($|G| > |K|$).  Agents will be assigned to one of the $K$ facilities and have preferences over the facility they are assigned to.  In particular their utility for being assigned to the facility at location $g \in G$ is $\theta_g = b - c_g$ where $b$ is the benefit of using a facility and $c_g$ is the cost associated with using the facility at location $g$.  We study the resulting mechanism design problem under the assumption that the mechanism can enforce the assignment of an agent to a particular facility, an assumption called {\em cluster imposing} in the literature~\cite{diodato2016}.  We can directly apply Theorem~\ref{thm:deterministic} to characterize the verification needed to ensure that all deterministic implementable-with-payments allocation rules are truthful with this verification.  At this level of abstraction, Figure~\ref{fig:det-mech-two-plus-zero}(a) captures the relevant pairwise constraints, and the overall harmless set is not substantively different than in our analysis of VCG (which uses an implementable-with-payments allocation rule for this problem) except that the null assignment is not permitted.

Our results become more interesting when we study the restricted case where the agents and possible locations are on a line and $c_g$ is the distance from the agent's location to $g$.  This setting was previously studied by Ferraioli et al.~(\citeyear{diodato2016}), who showed that in addition to the cluster imposing assumption, a combination of two (allocation dependent) verifications is sufficient to implement every {\em efficient} deterministic mechanism (with fixed tie-breaking).  The first, {\em no underbidding}, ensures the agent cannot report that she is closer to her assigned facility than she actually is.  The second, {\em direction imposing}, ensures the agent cannot report she is to the left of her assigned facility when she is actually to the right (and vice versa).
Because agent locations are restricted to be on the line, agent types are quite restricted.  When restricting to the pairwise case, if (WLOG) the agent prefers the right location, the harmless set for all implementable-with-payments allocation rules consists of all types to the left of the agent along the line.  If the agent's location is in between the two possible facility locations, then their two verifications exactly cover the complement of the harmless set: no underbidding prevents reports to the right of the agent's location but left of the facility location while direction imposing prevents reports to the right of the facility location.  If the agent is located to the right of both facilities, neither verification prevents misreports further to the right.  Instead, the restriction to allocation rules which use fixed tie-breaking ensures that these reports never change the allocation, so the harmless set in this case is actually the entire space.  

In addition to providing an intuitive illustration of why their verifications are sufficient (and in a sense necessary as well), we can strengthen their characterization to cover a larger class of mechanisms.  In particular, let a {\em fixed tie-breaking implementable-with-payments allocation rule} be an implementable-with-payments allocation rule with the additional property that all types which are indifferent between two allocations at prices implied by the allocation rule receive the same allocation.
\begin{corollary}
In the cluster imposing case, the no underbidding and direction imposing verifications suffice to implement all (efficient and approximate) fixed tie-breaking implementable-with-payments allocation rules.
\end{corollary}
We can also shed more light on whether their verifications are necessary.  They show that eliminating any one of them breaks truthfulness, which our results succintly illustrated.  However, their verifications are stronger than necessary in that they are still applied in the case where the agent would already receive her preferred allocation by reporting truthfully (and so the harmless set is the entire space).  So in principle the verifications could be weakened to no-underbidding-when-not-receiving-preferred-allocation and direction-imposing-when-not-receiving-preferred-allocation respectively.

\section{Conclusion}

We have introduced a general methodology that can be used to identify the harmless set of types. Knowing this set helps a mechanism designer identify the assumptions needed in order to use partial verification as a substitute for money in a new domain. We have pointed out that the power of verification depends on the class of mechanisms considered: in the case of deterministic and universally truthful mechanisms the harmless set of types is usually significantly larger than in case of truthful in expectation mechanisms. Furthermore, we discuss two possible extensions to our results: allocation-dependent verification and the reverse approach. We conclude with examples showing how our approach can be used in three application domains and how our results replicate and extend existing results in the literature.

\bibliographystyle{aaai}  
\bibliography{sigproc}

\end{document}